\documentclass[suppldata]{interact}\usepackage[]{graphicx}\usepackage[usenames,dvipsnames]{xcolor}
\makeatletter
\def\maxwidth{ %
  \ifdim\Gin@nat@width>\linewidth
    \linewidth
  \else
    \Gin@nat@width
  \fi
}
\makeatother

\definecolor{fgcolor}{rgb}{0.345, 0.345, 0.345}

\usepackage{framed}
\makeatletter
\newenvironment{kframe}{%
 \def\at@end@of@kframe{}%
 \ifinner\ifhmode%
  \def\at@end@of@kframe{\end{minipage}}%
  \begin{minipage}{\columnwidth}%
 \fi\fi%
 \def\FrameCommand##1{\hskip\@totalleftmargin \hskip-\fboxsep
 \colorbox{shadecolor}{##1}\hskip-\fboxsep
     \hskip-\linewidth \hskip-\@totalleftmargin \hskip\columnwidth}%
 \MakeFramed {\advance\hsize-\width
   \@totalleftmargin\z@ \linewidth\hsize
   \@setminipage}}%
 {\par\unskip\endMakeFramed%
 \at@end@of@kframe}
\makeatother

\definecolor{shadecolor}{rgb}{.97, .97, .97}
\definecolor{messagecolor}{rgb}{0, 0, 0}
\definecolor{warningcolor}{rgb}{1, 0, 1}
\definecolor{errorcolor}{rgb}{1, 0, 0}
\newenvironment{knitrout}{}{} 

\usepackage{alltt}
\usepackage{multirow}
\usepackage[]{graphicx}
\usepackage[]{color}
\usepackage{algorithm,algpseudocode,float}
\usepackage{subcaption}
\makeatletter
\def\maxwidth{ %
	\ifdim\Gin@nat@width>\linewidth
	\linewidth
	\else
	\Gin@nat@width
	\fi
}
\makeatother

\definecolor{fgcolor}{rgb}{0.345, 0.345, 0.345}

\usepackage{framed}
\makeatletter
{\par\unskip\endMakeFramed%
	\at@end@of@kframe}
\makeatother

\definecolor{shadecolor}{rgb}{.97, .97, .97}
\definecolor{messagecolor}{rgb}{0, 0, 0}
\definecolor{warningcolor}{rgb}{1, 0, 1}
\definecolor{errorcolor}{rgb}{1, 0, 0}

\usepackage{alltt}
\usepackage{amsmath,amsthm,mathtools}
\usepackage{graphics}
\usepackage{algorithm}
\usepackage{algpseudocode}
\usepackage{color,soul}
\usepackage{float}
\usepackage{multicol}
\usepackage[makeroom]{cancel}
\usepackage{comment}
\usepackage[usenames,dvipsnames]{xcolor}
\usepackage[justification=justified,skip=2pt ]{caption}
\usepackage{epstopdf}
\usepackage{natbib}

\renewcommand\bibfont{\fontsize{10}{12}\selectfont}
\makeatletter
\def\NAT@def@citea{\def\@citea{\NAT@separator}}
\makeatother
\theoremstyle{plain}
\newtheorem{theorem}{Theorem}[section]
\newtheorem{lemma}[theorem]{Lemma}

\newtheorem{proposition}[theorem]{Proposition}

\theoremstyle{definition}

\newtheorem{example}[theorem]{Example}

\theoremstyle{remark}

\IfFileExists{upquote.sty}{\usepackage{upquote}}{}

\IfFileExists{upquote.sty}{\usepackage{upquote}}{}
\begin{document}
	
	\title{Bias Reduction for Local Polynomial Derivative Estimation}
	
\author{
\name{
Fujia Chang\textsuperscript{a}
\thanks{CONTACT Fujia Chang. Email: yyaoo@mail.ubc.ca}
and W. John Braun\textsuperscript{a}
}
\affil{
\textsuperscript{a}Department of Computer Science, Mathematics, Physics and Statistics,
University of British Columbia, Kelowna, BC, Canada
}
}

\maketitle
	
\begin{abstract}
Local polynomial smoothing is commonly used in non-parametric regression, but local linear derivative estimation still has a bias of order $O(h^2)$. This paper proposes an iterative data sharpening method to reduce the bias of derivative estimates while retaining the simplicity of local linear fitting. The method is based on two expectation operators: $L_0$, acting on the regression function, and $L_1$, acting on the first-order derivative. By repeatedly applying the residual operator $R = I - L_0$, a series of sharpened derivative estimates can be constructed. After $l$ sharpening steps, the bias order can be reduced from $O(h^2)$ to $O(h^{2l+2})$. For the Gaussian kernel, all sharpening coefficients equal $1$, giving a simple closed-form single-bandwidth expression. Simulation experiments on three smooth test functions show that this method can significantly reduce the estimated bias, while revealing a bias--variance trade-off.
\end{abstract}

\section{Introduction}

In nonparametric regression, the local polynomial method is a standard tool for using noisy data to estimate unknown regression functions. Among the various available methods, local linear estimation may be the most widely used one. At an interior point, the order of bias is $O(h^2)$, and the order of variance is $O((nh)^{-1})$, so it can be used as a convenient reference smoother \citep{fan2018local, gasser1984estimating, buja1989linear}. A natural thought is how to reduce the smoothing bias without using more complex estimators. 
The idea of data sharpening was first developed in density estimation by \citet{choiHallDensity}, and was later applied to nonparametric regression by \citet{choi2000data}.
The basic idea is very simple: instead of changing the smoother itself, one just makes a small adjustment to the design points or response values, and then feeds the adjusted data into a smoother, such as the Nadaraya--Watson estimator \citep{nadaraya1964estimating, watson1964smooth} or the local linear estimator. The dominant term in the asymptotic expansion for the  bias can then be removed by choosing an appropriate adjustment while keeping the original smoother unchanged. 
Since then, this idea has also been extended to other nonparametric problems, including high-order data sharpening for density estimation \citep{hallKangHighOrder}, unimodal kernel density estimation \citep{hall2005unimodal}, bandwidth selection and related refinements for sharpened regression estimators \citep{naitoSharpening}, and high-order sharpening under dependent errors \citep{heSharpening}.
It has been used for unimodal density estimation under interval censoring \citep{becker2017interval}, and is also used for local likelihood models. 
In particular, \citet{BRAUN2020108831} developed data sharpening under the local likelihood framework through the Firth bias-reducing score adjustment \citep{firth1993bias}, and their framework also covers derivative estimation.

Reduction of bias can also be done through a combination of estimators obtained using different bandwidths. In nonparametric as well as semiparametric regression, \citet{cheng2018bias} proposed a multi-bandwidth bias reduction method. The method combines local linear estimators under different bandwidths to reduce bias. In this way, one can construct a bias-reduced estimator without necessarily having to fit higher-order local polynomials. 
Recently, \citet{Chen02012025} proposed an iterative data sharpening method for estimating the regression mean function.
With the local constant and local linear smoothers, this method can reduce the bias in the interior of the design range to order $O(h^{2M+2})$ after $M$ steps.

Derivative estimation remains different and is usually more sensitive.
Moving from mean estimation to derivative estimation is not just a change of estimation target. The bias expansion of the derivative estimator is different, and its variance is also larger. For example, at an interior point, the local linear estimator of the first-order derivative typically has variance of order $O((nh^3)^{-1})$, not $O((nh)^{-1})$ \citep{fan2018local}. This shows that the bias--variance trade-off in derivative estimation is more important. One way to reduce the derivative bias is to use higher-order local polynomials. However, higher-order derivative estimators can be more sensitive to bandwidth selection and may not be stable enough in finite samples, especially near the boundary \citep{fan2018local, de2013derivative, li2003multivariate}. 
In principle, a multi-bandwidth method similar to \citet{cheng2018bias} could also be adapted to derivative estimation, for example by combining local linear slope estimates under different bandwidths.
However, such an approach requires fitting estimators at several bandwidths and choosing a bandwidth grid.
It is therefore meaningful to study a bias reduction method that uses only a single bandwidth and at the same time keeps the simplicity of the local polynomial derivative estimator.

Motivated by these existing approaches, this paper proposes an iterative data sharpening method for local linear derivative estimation.
The method uses two expectation operators. The first operator $L_0$ acts on the regression function. The second operator $L_1$ acts on the first-order derivative of the regression function. We repeatedly apply the residual operator $R = I - L_0$ to construct the sharpened function, and then use $L_1$ to define the derivative estimator. 
Unlike the local likelihood construction of \citet{BRAUN2020108831}, the proposed method uses the iterative residual operator $R=I-L_0$; unlike multi-bandwidth approaches, it repeatedly uses the same local linear smoother with a single bandwidth.
Then the dominant derivative bias can be gradually cancelled. After $L$ sharpening steps, the bias can be reduced from $O(h^2)$ to $O(h^{2L+2})$. For the Gaussian kernel, the sharpening coefficient at each step is equal to $1$, so a simple closed-form expression can be obtained.

We tested the method's performance in finite samples via a simulation study. The simulation study uses a smooth sine function and a polynomial function, and demonstrates the effects of sharpening on the bias, standard deviation and integrated mean squared error of the derivative estimator. Since the goal is to reduce bias, the amount of variance introduced by the method also must be checked. Lastly, we apply the approach to the motorcycle impact data. Here the true derivative is unknown, the primary objective is to demonstrate how the derivative estimate curve changes after sharpening.
\section{Background}

\subsection{Local Polynomial Regression Derivative Estimation}
 
Take independent observations $\left(x_1, y_1\right), \ldots,\left(x_n, y_n\right)$ from the regression model
$$
y_i=g\left(x_i\right)+\varepsilon_i, \quad i=1, \ldots, n,
$$
where $g(x)$ is an unknown smooth regression function and the $\varepsilon_i$ are independent errors with zero mean and finite variance. The target is $g^{(r)}(x)$, the $r$-th derivative of $g$ at some fixed point $x \in \mathbb{R}$.

The local polynomial approach fits a Taylor approximation of $g(x_i)$ near $x$,
$$
g\left(x_i\right) \approx a_0+a_1\left(x_i-x\right)+\cdots+a_p\left(x_i-x\right)^p,
$$
with $p \geq r$, and reads off the relevant coefficient. The vector $\boldsymbol{a}=\left(a_0, \ldots, a_p\right)^{\top}$ is the minimiser of the weighted least squares 
$$
\min _a \sum_{i=1}^n\left[y_i-\sum_{j=0}^p a_j\left(x_i-x\right)^j\right]^2 K_h\left(x_i-x\right),
$$
where $K_h(u)=\frac{1}{h} K\left(\frac{u}{h}\right)$ is a symmetric kernel with bandwidth $h$ that puts more weight on observations close to $x$.

In matrix form, write
$$
\mathbf{X}=\left[\begin{array}{cccc}
1 & x_1-x & \cdots & \left(x_1-x\right)^p \\
\vdots & \vdots & & \vdots \\
1 & x_n-x & \cdots & \left(x_n-x\right)^p
\end{array}\right], \quad \mathbf{W}=\operatorname{diag}\left(K_h\left(x_1-x\right), \ldots, K_h\left(x_n-x\right)\right) ,
$$
so that by taking the derivative,
$$
\hat{\boldsymbol{a}}=\left(\mathbf{X}^{\top} \mathbf{W} \mathbf{X}\right)^{-1} \mathbf{X}^{\top} \mathbf{W} \mathbf{y} .
$$
The local polynomial estimator of $g(x)$ is the first entry of $\hat{\boldsymbol{a}}$,
$$
\hat{g}(x)=\hat{a}_0=\mathbf{e}_0^{\top} \hat{\boldsymbol{a}},
$$
with $\mathbf{e}_0=(1,0, \ldots, 0)^{\top} \in \mathbb{R}^{p+1}$, and the $r$-th derivative estimator picks out $\hat{a}_r$ scaled by $r!$:
$$
\hat{g}^{(r)}(x)=r!\cdot \hat{a}_r=r!\cdot \mathbf{e}_r^{\top} \hat{\boldsymbol{a}},
$$
where $\mathbf{e}_r$ is the standard basis vector with a one in the $(r+1)$-th position.

\subsection{Bias Reduction Strategies} \label{theory}
Several bias reduction methods have been studied. \citet{cheng2018bias} proposed a multi-bandwidth bias reduction method  for estimating the regression function itself. Their method combines local linear estimates of $g(x_0)$ obtained under several bandwidths and extrapolates their relationship with $h^2$ to $h^2=0$. They did not consider derivative estimation. For comparison in this paper, we adapt their bandwidth-extrapolation idea to local linear first-derivative estimation.

Assume
\[
Y_i = g(X_i) + \epsilon_i,
\]
where $g$ is twice differentiable, the errors $\epsilon_i$ are mutually independent, with mean zero and finite variance. In order to estimate $g'(x_0)$, a local linear regression in the neighborhood of $x_0$ is fitted:
\[
Y_i \approx a + b(X_i - x_0).
\]
The slope $b$ is the local linear derivative estimate, written as $\hat{g}'_h(x_0)$. At an interior point, the order of its dominant bias is $O(h^2)$.
The multi-bandwidth correction method computes the local linear slope estimate at the same point $x_0$, but uses several different bandwidths $h_1,\ldots,h_B$. From this, a set of slope estimates can be obtained:
\[
b_i = \hat{g}'_{h_i}(x_0), \qquad i = 1, \ldots, B.
\]
Since the dominant bias is proportional to $h_i^2$, these estimates can be approximated by the regression model:
\[
b_i = \alpha + \beta h_i^2 + e_i.
\]
Then, the estimated intercept $\hat{\alpha}$ is viewed as the bias-corrected derivative estimate:
\[
\hat{g}'_{\mathrm{BC}}(x_0) = \hat{\alpha}.
\]
The basic idea of this method is to first fit the relationship between the slope estimate and $h^2$, and then set $h^2=0$, so as to remove the dominant $O(h^2)$ bias term. This idea is, in fact, very simple, but it requires repeated computation under different bandwidths. In addition, the choice of bandwidth also affects the performance.

\section{Iterated Data Sharpening for Derivative Estimation}\label{sec:iterated_ds}

The data sharpening method for regression function estimation was originally proposed by \citet{choi2000data} as a bias reduction method. \citet{Chen02012025} then develops it into an iterative procedure for nonparametric regression estimation, while \citet{BRAUN2020108831} introduced sharpening ideas related to first-order derivative estimation. The following method extends the idea of derivative sharpening to an iterative framework. We first discuss the general kernel. The Gaussian kernel gives a simpler closed form, which is discussed later.

\subsection{General framework}

In this section, $x$ is regarded as an interior point of the support. We assume that $K$ is a symmetric kernel and that $g$ is sufficiently smooth near $x$.

Let $\hat g(x)$ and $\hat g'(x)$ denote the local linear estimators of the regression function and its first derivative at $x$. Define
\[
L_0[g](x) := \mathbb{E}[\hat g(x)],
\qquad
L_1[g](x) := \mathbb{E}[\hat g'(x)].
\]

At an interior point, the local linear estimators can be written in equivalent-kernel form. For a symmetric kernel $K$, we write
\[
L_0[g](x) = \int K(u)\,g(x+hu)\,du,
\qquad
L_1[g](x) = \frac{1}{\mu_2 h}\int u\,K(u)\,g(x+hu)\,du,
\]
where
\[
\mu_j := \int u^j K(u)\,du
\]
is the $j$th moment of $K$.

Using the Taylor expansion of $g(x+hu)$ around $x$,
\[
g(x+hu) = \sum_{r=0}^\infty \frac{(hu)^r}{r!}\,g^{(r)}(x),
\]
we obtain
\[
L_0[g](x)
= \int K(u)\sum_{r=0}^\infty \frac{(hu)^r}{r!}\,g^{(r)}(x)\,du
= \sum_{r=0}^\infty \frac{h^r}{r!}\,g^{(r)}(x)\int u^r K(u)\,du.
\]
Since $K$ is symmetric, all odd moments vanish, so only the even-order terms remain:
\[
L_0[g](x)
= g(x) + \sum_{k=1}^\infty \frac{\mu_{2k}}{(2k)!}\,h^{2k}\,g^{(2k)}(x).
\]
Set
\[
a_k := \frac{\mu_{2k}}{(2k)!},
\qquad k \ge 1.
\]

A similar calculation gives the expansion of $L_1[g](x)$. We have
\[
L_1[g](x)
= \frac{1}{\mu_2 h}\int u\,K(u)\sum_{r=0}^\infty \frac{(hu)^r}{r!}\,g^{(r)}(x)\,du
= \frac{1}{\mu_2 h}\sum_{r=0}^\infty \frac{h^r}{r!}\,g^{(r)}(x)\int u^{r+1}K(u)\,du.
\]
In this case, the symmetry of $K$ leaves only the terms with odd-order derivatives of $g$. Taking $r = 2k+1$, we get
\[
L_1[g](x)
= g'(x) + \sum_{k=1}^\infty \frac{\mu_{2k+2}}{(2k+1)!\,\mu_2}\,h^{2k}\,g^{(2k+1)}(x).
\]
Set
\[
b_k := \frac{\mu_{2k+2}}{(2k+1)!\,\mu_2},
\qquad k \ge 1.
\]

Thus $L_0$ and $L_1$ have the formal expansions
\begin{equation}
L_0[g] = g + \sum_{k=1}^\infty a_k\, h^{2k}\, g^{(2k)},
\label{eq:L0_general}
\end{equation}
and
\begin{equation}
L_1[g] = g' + \sum_{k=1}^\infty b_k\, h^{2k}\, g^{(2k+1)}.
\label{eq:L1_general}
\end{equation}

Now define the residual operator
\[
R := I - L_0.
\]
From \eqref{eq:L0_general}, this gives
\begin{equation}
Rg = -\sum_{k=1}^\infty a_k\, h^{2k}\, g^{(2k)}.
\label{eq:R_expand}
\end{equation}

The sharpened functions are constructed recursively as
\begin{equation}
g_0 := g,
\qquad
g_l := g_{l-1} + \alpha_l\, R^l g,
\qquad l \ge 1.
\label{eq:gm_recursive}
\end{equation}
Here the constants $\alpha_1, \alpha_2, \dots$ are chosen so that one additional leading bias term is removed at each step.

We first need two simple lemmas.

\begin{lemma}\label{lem:Rmg}
For any integer $l \ge 1$,
\[
R^l g = (-a_1)^l\, h^{2l}\, g^{(2l)} + O(h^{2l+2}).
\]
\end{lemma}

\begin{proof}
We prove the result by induction on $l$.

For $l = 1$, \eqref{eq:R_expand} gives
\[
Rg
= -a_1 h^2 g^{(2)} - \sum_{k=2}^\infty a_k\, h^{2k}\, g^{(2k)}
= -a_1 h^2 g^{(2)} + O(h^4).
\]
So the result holds for $l = 1$.

Now suppose the result holds for some $l \ge 1$. That is,
\[
R^l g = (-a_1)^l\, h^{2l}\, g^{(2l)} + O(h^{2l+2}).
\]
Applying $R$ once more and using \eqref{eq:R_expand}, we obtain
\[
R^{l+1} g
= R(R^l g)
= -\sum_{k=1}^\infty a_k\, h^{2k}\, (R^l g)^{(2k)}.
\]
The leading term comes from $k = 1$:
\[
-a_1 h^2 (R^l g)^{(2)}
= -a_1 h^2 \left[ (-a_1)^l\, h^{2l}\, g^{(2l+2)} + O(h^{2l+2}) \right].
\]
All terms with $k \ge 2$ are of higher order because they contain at least a factor $h^4$. Therefore,
\[
R^{l+1} g = (-a_1)^{l+1}\, h^{2l+2}\, g^{(2l+2)} + O(h^{2l+4}).
\]
This proves the result for $l+1$, and the induction argument is complete.
\end{proof}

\begin{lemma}\label{lem:L1Rmg}
For any integer $l \ge 1$,
\[
L_1[R^l g] = (-a_1)^l\, h^{2l}\, g^{(2l+1)} + O(h^{2l+2}).
\]
\end{lemma}

\begin{proof}
Using \eqref{eq:L1_general},
\[
L_1[R^l g] = (R^l g)' + \sum_{k=1}^\infty b_k\, h^{2k}\, (R^l g)^{(2k+1)}.
\]
By Lemma~\ref{lem:Rmg},
\[
R^l g = (-a_1)^l\, h^{2l}\, g^{(2l)} + O(h^{2l+2}).
\]
Differentiating with respect to $x$ gives
\[
(R^l g)' = (-a_1)^l\, h^{2l}\, g^{(2l+1)} + O(h^{2l+2}).
\]
The remaining sum has an extra factor of $h^2$, so it is of order $O(h^{2l+2})$. Hence,
\[
L_1[R^l g] = (-a_1)^l\, h^{2l}\, g^{(2l+1)} + O(h^{2l+2}).
\]
\end{proof}

We also use the fact that each sharpened derivative estimator has an expansion involving only odd-order derivatives of $g$.

\begin{lemma}\label{lem:L1gm_form}
For each $l \ge 0$, there exist constants $C_{l,q}$ such that
\[
L_1[g_l] = g' + \sum_{q=1}^\infty C_{l,q}\, h^{2q}\, g^{(2q+1)}.
\]
\end{lemma}

\begin{proof}
From the recursive definition in \eqref{eq:gm_recursive},
\[
g_l = g + \sum_{j=1}^l \alpha_j\, R^j g.
\]
By Lemma~\ref{lem:Rmg}, each $R^j g$ is a linear combination of terms of the form
\[
h^{2q}\, g^{(2q)}.
\]
Applying $L_1$ to these terms and using \eqref{eq:L1_general} gives only terms of the form
\[
h^{2q}\, g^{(2q+1)}.
\]
This proves the claimed form of the expansion.
\end{proof}

We can now state the main result for a general symmetric kernel.

\begin{theorem}\label{thm:general_iterative}
Suppose $a_1 \neq 0$. Then, for every $l \ge 1$, there exist constants $\alpha_1, \dots, \alpha_l$ such that the recursively defined function $g_l$ in \eqref{eq:gm_recursive} satisfies
\[
L_1[g_l](x) - g'(x) = O(h^{2l+2}).
\]
\end{theorem}

\begin{proof}
We prove the result by induction on $l$.

When $l = 0$, \eqref{eq:L1_general} gives
\[
L_1[g] - g' = O(h^2).
\]

Now suppose that $\alpha_1, \dots, \alpha_{l-1}$ have been chosen so that
\[
L_1[g_{l-1}] - g' = O(h^{2l}).
\]
By Lemma~\ref{lem:L1gm_form}, the first nonzero term has the form
\[
L_1[g_{l-1}] - g' = C_l\, h^{2l}\, g^{(2l+1)} + O(h^{2l+2})
\]
for some constant $C_l$. Since
\[
g_l = g_{l-1} + \alpha_l\, R^l g,
\]
we have
\[
L_1[g_l] - g' = (L_1[g_{l-1}] - g') + \alpha_l\, L_1[R^l g].
\]
Using Lemma~\ref{lem:L1Rmg}, this becomes
\[
L_1[g_l] - g' = \left( C_l + \alpha_l\,(-a_1)^l \right) h^{2l}\, g^{(2l+1)} + O(h^{2l+2}).
\]
Choosing
\[
\alpha_l = -\frac{C_l}{(-a_1)^l}
\]
removes the leading $h^{2l}$ term. Therefore,
\[
L_1[g_l] - g' = O(h^{2l+2}).
\]
This completes the induction proof.
\end{proof}

Theorem~\ref{thm:general_iterative} only proves the existence of suitable sharpening coefficients, but the derivation above shows that they can be computed for any symmetric kernel, giving an explicit formula expressed in terms of the even-order moments of the kernel. By tracking the coefficient of each $h^{2j}$ term in $L_1[g_l]-g'$ during the construction in~\eqref{eq:gm_recursive}, we obtain, for example,
\begin{align}
\alpha_1 &= \frac{\mu_4}{3\,\mu_2^2}, \label{eq:alpha1_closed}\\
\alpha_2 &= \frac{5\,\mu_4^2 - \mu_2 \mu_6}{30\,\mu_2^4}, \label{eq:alpha2_closed}\\
\alpha_3 &= \frac{3\,\mu_2^2 \mu_8 - 70\,\mu_2 \mu_4 \mu_6 + 175\,\mu_4^3}{1890\,\mu_2^6}. \label{eq:alpha3_closed}
\end{align}
These expressions hold for any symmetric kernel. For the Gaussian kernel, they simplify to $\alpha_l=1$. When $l\geq 4$, the corresponding coefficients can be derived by the same method, but their closed-form expressions become increasingly complex. In practice, these coefficients can be obtained by further numerical recursion; see the supplementary code (\texttt{ds\_alphas.R}) for the specific implementation.

\begin{example}[Sharpening coefficients for the Epanechnikov kernel]\label{ex:epanechnikov}
Consider the Epanechnikov kernel, defined by
\[
K(u)=\frac{3}{4}(1-u^2), \qquad |u|\leq 1,
\]
and $K(u)=0$ for $|u|>1$. Direct integration gives the kernel moments
\[
\mu_2 = \frac{1}{5}, \qquad \mu_4 = \frac{3}{35}, \qquad \mu_6 = \frac{1}{21}, \qquad \mu_8 = \frac{1}{33}.
\]
Substituting these into~\eqref{eq:alpha1_closed}--\eqref{eq:alpha3_closed} gives
\[
\alpha_1 = \frac{5}{7} \approx 0.714, \qquad \alpha_2 = \frac{250}{441} \approx 0.567, \qquad \alpha_3 = \frac{47750}{101871} \approx 0.469.
\]
\end{example}

Table ~\ref{tab:kernel_alphas} provides the first three sharpening coefficients for five commonly used symmetric kernels. Unlike the Gaussian kernel discussed in the next subsection, these coefficients are not all equal to 1, and they are also clearly different from each other.

\begin{table}[ht]
\centering
\caption{Sharpening coefficients $\alpha_1, \alpha_2, \alpha_3$ for common kernels}
\label{tab:kernel_alphas}
\begin{tabular}{lccc}
\hline
Kernel & $\alpha_1$ & $\alpha_2$ & $\alpha_3$ \\
\hline
Gaussian      & $1$                   & $1$       & $1$       \\
Uniform       & $3/5 = 0.6000$        & $0.4114$  & $0.2971$  \\
Epanechnikov  & $5/7 \approx 0.7143$  & $0.5669$  & $0.4687$  \\
Biweight      & $7/9 \approx 0.7778$  & $0.6599$  & $0.5791$  \\
Triweight     & $9/11 \approx 0.8182$ & $0.7209$  & $0.6535$  \\
\hline
\end{tabular}
\end{table}

\subsection{The Gaussian kernel case}\label{sec:gaussian_case}

The Gaussian kernel makes things considerably cleaner. For the standard Gaussian,
\[
\mu_{2k}=\frac{(2k)!}{2^k k!},
\qquad \mu_2=1,
\]
so \eqref{eq:L0_general} and \eqref{eq:L1_general} specialise to
\begin{equation}
L_0[g]
=
g+\sum_{k=1}^\infty \frac{1}{2^k k!} h^{2k} g^{(2k)},
\label{eq:L0_gaussian}
\end{equation}
and
\begin{equation}
L_1[g]
=
g'+\sum_{k=1}^\infty \frac{1}{2^k k!} h^{2k} g^{(2k+1)}.
\label{eq:L1_gaussian}
\end{equation}
Both the $a_k$ and the $b_k$ collapse to the same value:
\[
a_k=b_k=\frac{1}{2^k k!},
\qquad k\ge1.
\]

It turns out that all the sharpening coefficients are then equal to one, so the update in \eqref{eq:gm_recursive} simplifies to
\[
g_0:=g,
\qquad
g_l:=g_{l-1}+R^l g,
\qquad l\ge1,
\]
with \(R=I-L_0\), or, writing things out,
\[
g_l=\sum_{j=0}^l R^j g.
\]
This makes the sharpened derivative estimator explicit.

\begin{proposition}\label{prop:gaussian_explicit}
For every \(l\ge1\),
\begin{equation}
L_1[g_l](x)
=
g'(x)
+
(-1)^l
\sum_{k_l=1}^{\infty}\cdots\sum_{k_0=1}^{\infty}
\frac{h^{2\sum_{j=0}^l k_j}}
{2^{\sum_{j=0}^l k_j}\,k_0!\cdots k_l!}
\,g^{\left(2\sum_{j=0}^l k_j+1\right)}(x).
\label{eq:gaussian_explicit_main}
\end{equation}
Hence
\[
L_1[g_l](x)-g'(x)=O(h^{2l+2}).
\]
\end{proposition}

\begin{proof}[Proof sketch]
Induction on \(l\). For \(l=1\),
\[
g_1=g+Rg=g+(I-L_0)g,
\]
and substituting the Gaussian series expansions of \(L_0\) and \(L_1\) leaves
\[
L_1[g_1](x)
=
g'(x)
-
\sum_{k_1=1}^{\infty}\sum_{k_0=1}^{\infty}
\frac{h^{2(k_0+k_1)}}
{2^{k_0+k_1}k_0!k_1!}
\,g^{(2(k_0+k_1)+1)}(x),
\]
which is \eqref{eq:gaussian_explicit_main} at \(l=1\).

For the induction step, assume \eqref{eq:gaussian_explicit_main} at level \(l\) and use
\[
g_{l+1}=g_l+R^{l+1}g
=
g_l+(I-L_0)^{l+1}g.
\]
Expanding \((I-L_0)^{l+1}\) by the binomial theorem and applying \(L_1\) gives a double-indexed sum that can be grouped by the number of positive indices, and the remaining cross-terms cancel through the Riordan identity
\[
\sum_{j=0}^{n}(-1)^j\binom{n}{j}\binom{j}{k}=0,
\qquad k<n,
\]
which yields \eqref{eq:gaussian_explicit_main} at level \(l+1\).

The smallest value of \(\sum_{j=0}^l k_j\) on the right-hand side of \eqref{eq:gaussian_explicit_main} is \(l+1\), so the leading term is of order \(h^{2l+2}\) and
\[
L_1[g_l](x)-g'(x)=O(h^{2l+2}).
\]
The full combinatorial argument is provided in Appendix~\ref{app:gaussian_combinatorial}.
\end{proof}

So the Gaussian kernel is genuinely special in this framework. For a general kernel the \(\alpha_l\) come out of a recursive cancellation and typically do not admit a simple closed form, but for the Gaussian kernel they collapse to
\[
\alpha_l=1,
\qquad l\ge1.
\]

\subsection{Implications for the variance}

The above results mainly focus on the bias of $L_1[g_l]$. In the actual computation, the estimator acts on the observed responses, so the noise part also needs to be considered. Therefore, sharpening can reduce bias, but it may also change the variance.

For a fixed target point $x$, the order-$l$ sharpened derivative estimator is still a linear estimator of the responses. We write it as
\[
\hat g_l'(x)
= \sum_{i=1}^n w_{l,i}(x) Y_i,
\]
where $w_{l,i}(x)$ is the effective derivative weight after sharpening. These weights depend on the bandwidth, the kernel, the design points, and the sharpening order. When $l = 0$, they are the weights of the ordinary local linear derivative estimator.

Assume
\[
Y_i = g(x_i) + \varepsilon_i,
\qquad
E(\varepsilon_i \mid X) = 0,
\qquad
\operatorname{Var}(\varepsilon_i \mid X) = \sigma^2,
\]
and the error terms are conditionally independent given $X$. Then
\[
\hat g_l'(x)
= \sum_{i=1}^n w_{l,i}(x) g(x_i)
+ \sum_{i=1}^n w_{l,i}(x) \varepsilon_i.
\]
Given the design points, the first term is fixed. Therefore,
\[
\operatorname{Var}\{\hat g_l'(x) \mid X\}
= \operatorname{Var}\!\left\{
\sum_{i=1}^n w_{l,i}(x) \varepsilon_i
\;\middle|\; X
\right\}
= \sigma^2 \sum_{i=1}^n w_{l,i}(x)^2.
\]
Thus, sharpening affects the variance through these effective weights.

The following is an explanation of the usual order of this variance. At an interior point, if the design points are relatively regular, about $nh$ observation points will participate in the local fitting. For local mean estimation, each weight is usually of order $1/(nh)$. For first-order derivative estimation, because the derivative is estimated on the local scale $x_i - x = O(h)$, the weight has an additional factor of $1/h$. Therefore, the order of each ordinary derivative weight is
\[
O\!\left(\frac{1}{nh^2}\right).
\]
Since there are about $nh$ observation points involved in the estimation, the order of the squared sum of weights is
\[
(nh)
\left(
\frac{1}{nh^2}
\right)^2
= \frac{1}{nh^3}.
\]
So, the variance order of the ordinary local linear derivative estimator is
\[
\operatorname{Var}\{\hat g_0'(x) \mid X\}
= O\!\left(\frac{1}{nh^3}\right).
\]
For a fixed sharpening order $l$, the derivative weights are replaced by $w_{l,i}(x)$. However, under the same fixed-design setting, sharpening does not change the basic order of the variance with respect to the bandwidth. Therefore, the sharpened derivative estimator has the same order of variance,
\[
\operatorname{Var}\left\{\hat{g}_l^{\prime}(x) \mid X\right\}
= O\!\left(\frac{1}{n h^3}\right),
\qquad l \text{ fixed},
\]
where the leading constant may depend on $l$, the sharpening coefficients, and the kernel.
Hence, although the proposed sharpening method  reduces the bias from $O(h^2)$ to $O(h^{2l+2})$, it does not directly reduce the variance. This also explains why a moderate sharpening order can improve the IMSE, whereas a higher order, although it may further reduce the bias, may also increase the variance. This point will be illustrated in the simulations below.

\section{Simulation Research and Case Analysis}

\subsection{Simulation design}

We examine the performance of the proposed method in finite samples through simulation. This simulation study has two main goals. First, we hope to test whether the bias order given in Theorem~\ref{thm:general_iterative} can be reflected in the numerical results. Second, we would like to compare the sharpened estimate with the multi-bandwidth derivative benchmark adapted from the bandwidth-extrapolation idea of \citet{cheng2018bias}.

The data are generated from the following model:
\begin{equation*}
Y_i = g(x_i) + \varepsilon_i, \qquad \varepsilon_i \sim \mathcal{N}(0, \sigma^2),
\quad i = 1, \ldots, n,
\end{equation*}
where $n = 601$ and the design points are equally spaced on a given interval. We use three test functions to cover different smoothness situations. The first is $g_1(x) = 25 \sin(x/30) + 5$, defined on $[0, 300]$, which is a slowly changing smooth sine function. The second is $g_2(x) = x^7$, defined on $[-1, 1]$, which is a polynomial function with a higher but finite order derivative. The third is $g_3(x) = \sin(2\pi x)$, defined on $[0, 1]$, which is a smooth function with stronger local curvature. To reduce the impact of boundary effects, we only measure the estimated performance in the interior region of each function: $[60, 240]$ for $g_1$, $[-0.6, 0.6]$ for $g_2$, and $[0.15, 0.85]$ for $g_3$.

We compare six first-order derivative estimates. The first one is written as LL, that is, an ordinary local linear derivative estimate, whose theoretical bias order is $O(h^2)$. The next four are written as SH1--SH4, corresponding to the sharpening orders $l = 1, 2, 3, 4$ respectively, with theoretical bias order $O(h^{2l+2})$. The last one is denoted by MB and is the derivative adaptation of the multi-bandwidth bias reduction idea of \citet{cheng2018bias}. As described in Section~\ref{theory}, the MB method fits the local linear slope at five bandwidths $\{0.6, 0.8, 1.0, 1.2, 1.4\} \times h$ respectively, then extrapolates to $h^2 = 0$ by least squares. Its theoretical bias order is $O(h^4)$, the same as SH1, so it is a natural comparison object.

Throughout, we use the Gaussian kernel. As shown in Proposition~\ref{prop:gaussian_explicit}, all sharpening coefficients $\alpha_j$ are equal to one. Each method is evaluated on the same set of bandwidth grids, and we report the results under its optimal bandwidth, where the optimal bandwidth is defined as the bandwidth that minimizes IMSE. The Monte Carlo simulation uses $M = 500$ repetitions, and the noise level is set to $\sigma = 0.3$.

\subsection{Verification of theoretical bias order}
\label{sec:slopes}

We first use a numerical experiment to check Theorem~\ref{thm:general_iterative}. We set $\sigma=0$, so there is no random noise in the data and the estimation error only comes from the bias of the estimator. We plot the average absolute bias against the bandwidth $h$ on a log-log scale. For each method, we then fit a straight line using the middle part of the bandwidth grid. If the leading asymptotic term is dominant, the fitted slope should be close to the theoretical bias order discussed above.

\begin{table}[ht]
\centering
\caption{The theoretical bias orders and the fitted log-log slopes of $|\text{Bias}|$ against $h$ when $\sigma=0$.}
\label{tab:slopes}
\begin{tabular}{lcccc}
\hline
Method & Theoretical & $g_1$ & $g_2$ & $g_3$ \\
\hline
LL  & 2  & 1.70 & 2.03 & 2.09 \\
SH1 & 4  & 3.52 & 3.37 & 4.59 \\
SH2 & 6  & 5.46 & 4.11 & 6.00 \\
SH3 & 8  & 6.57 & 4.34 & 6.17 \\
SH4 & 10 & 5.95 & 4.37 & 6.34 \\
MB  & 4  & 3.80 & 3.92 & 3.60 \\
\hline
\end{tabular}
\end{table}

As shown in Table~\ref{tab:slopes}, the slopes for the lower-order methods are generally close to the theoretical results. For all three test functions, the slope of LL is close to $2$. The slopes of SH1 and MB are also reasonably close to the theoretical order $4$. In particular, for $g_3$, the measured slope of SH2 is $6.00$, which is very close to the theoretical value.

The agreement becomes weaker for the higher sharpening orders. For example, for $g_1$ and $g_3$, although the slopes of SH3 and SH4 are still clearly larger than those of LL and SH1, they do not reach the theoretical values of $8$ and $10$. In addition, for $g_1$, the slope of SH4 is slightly lower than that of SH3. In finite-sample numerical experiments, this phenomenon is reasonable. As the sharpening order increases, it can be seen from Figure~\ref{fig:bias} that the leading bias has become very small. At this point, the measured error may be affected not only by the leading asymptotic term, but also by higher-order remainder terms and other factors. Therefore, the slope obtained by fitting does not necessarily reflect only the leading term suggested by the theory. Theorem~\ref{thm:general_iterative} gives the pointwise asymptotic bias order, while the calculation in the simulation is the average absolute bias over the interior region. Therefore, the slope corresponding to the average absolute bias is not necessarily exactly the same as the pointwise theoretical order.

\begin{figure}[ht]
\centering
\includegraphics[width=\linewidth]{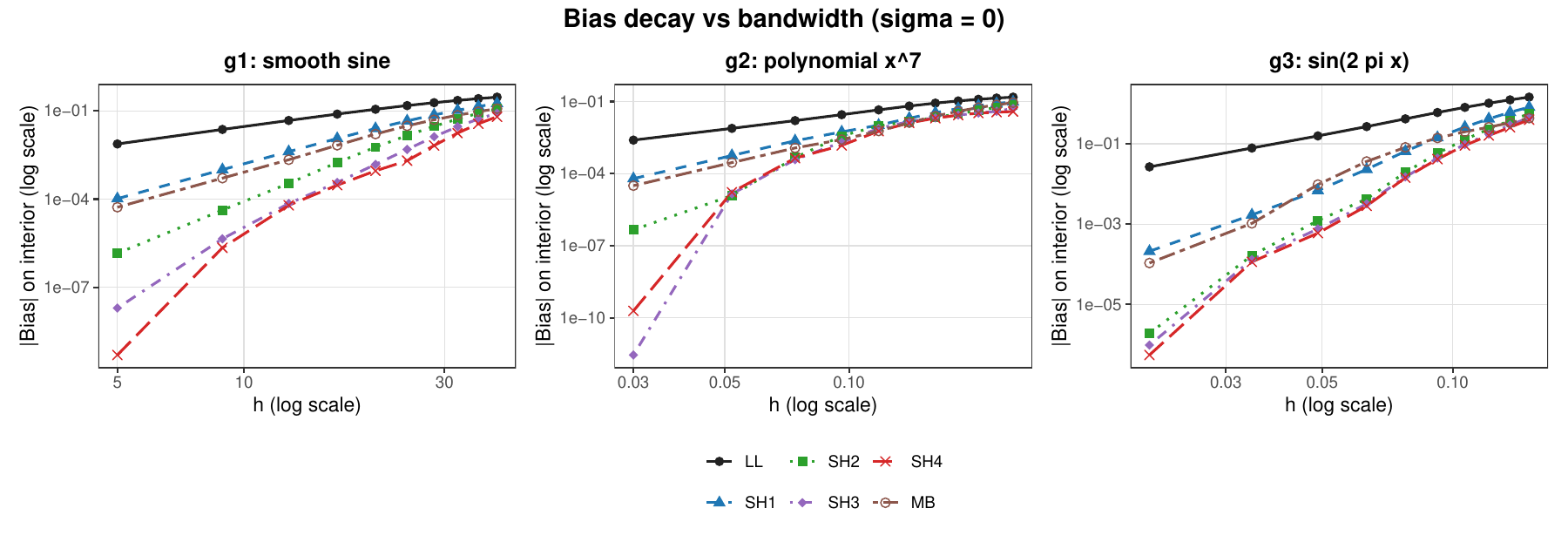}
\caption{Absolute bias as a function of bandwidth, with $\sigma=0$}
\label{fig:bias}
\end{figure}

The polynomial function $g_2(x)=x^7$ is also a special case. When $l\geq 3$,
\[
g_2^{(2l+2)}(x)\equiv 0,
\]
because the derivatives of order eight and above are equal to zero. Therefore, for SH3 and SH4, the nominal leading bias term in Theorem~\ref{thm:general_iterative} vanishes. At this point, the measured slopes should not be directly interpreted as numerical estimates of the theoretical orders $8$ and $10$. The remaining error is mainly determined by the discrete matrix calculation and numerical error. This also explains why the slopes of SH3 and SH4 are very close for $g_2$.

For SH2, the dominant term related to $g_2^{(6)}$ is still not zero. However, the measured slope is only $4.11$. This means that the bandwidth range currently in use may not have fully entered the asymptotic region dominated by the $h^6$ term. In order to get a slope closer to the theoretical value, it may be necessary to use a smaller bandwidth than the current lower bound of $0.03$.

In general, the numerical results clearly support the improvement of the bias order under the lower sharpening orders. For higher-order methods, the bias still decreases faster than for ordinary local linear estimates. However, under the current finite sample and finite bandwidth grid, it is more difficult to obtain a slope that is exactly consistent with the theoretical value.

\subsection{Finite-sample performance and bias--variance trade-off}
\label{sec:finite-sample}

The previous section verifies the order of bias without noise. Here we add noise ($\sigma=0.3$) so that bias and variance are both operative. For each method, we use its optimal bandwidth $h^*$, that is, the bandwidth that makes IMSE the smallest on the bandwidth grid. Table~\ref{tab:imse} gives the average absolute bias, standard deviation and IMSE under three test functions, and Figure~\ref{fig:imse} shows the change of IMSE with bandwidth.

\begin{table}[ht]
\centering
\caption{The best performance of each method over the bandwidth grid when $\sigma=0.3$.}
\label{tab:imse}
\small
\begin{tabular}{llcccc}
\hline
Function & Method & $h^*$ & Avg.\ $|\text{Bias}|$ & SD & IMSE \\
\hline
\multirow{6}{*}{$g_1(x)=25\sin(x/30)+5$}
 & LL  & 5.000  & $7.532\times10^{-3}$ & $7.097\times10^{-3}$ & $1.19\times10^{-4}$ \\
 & SH1 & 8.889  & $1.014\times10^{-3}$ & $4.436\times10^{-3}$ & $2.07\times10^{-5}$ \\
 & SH2 & 16.667 & $1.746\times10^{-3}$ & $2.083\times10^{-3}$ & $7.82\times10^{-6}$ \\
 & SH3 & 16.667 & $3.44\times10^{-4}$  & $2.345\times10^{-3}$ & $\mathbf{5.64\times10^{-6}}$ \\
 & SH4 & 20.556 & $8.95\times10^{-4}$  & $1.886\times10^{-3}$ & $5.97\times10^{-6}$ \\
 & MB  & 12.778 & $2.210\times10^{-3}$ & $3.750\times10^{-3}$ & $1.95\times10^{-5}$ \\
\hline
\multirow{6}{*}{$g_2(x)=x^7$}
 & LL  & 0.140 & 0.0671 & 0.1229 & 0.0248 \\
 & SH1 & 0.206 & 0.0688 & 0.1066 & 0.0197 \\
 & SH2 & 0.250 & 0.0726 & 0.0960 & 0.0169 \\
 & SH3 & 0.250 & 0.0512 & 0.1082 & $\mathbf{0.0153}$ \\
 & SH4 & 0.250 & 0.0368 & 0.1178 & 0.0158 \\
 & MB  & 0.228 & 0.0701 & 0.1329 & 0.0300 \\
\hline
\multirow{6}{*}{$g_3(x)=\sin(2\pi x)$}
 & LL  & 0.063 & 0.2699 & 0.2837 & 0.1753 \\
 & SH1 & 0.092 & 0.1466 & 0.2470 & $\mathbf{0.0923}$ \\
 & SH2 & 0.107 & 0.1325 & 0.2362 & 0.0967 \\
 & SH3 & 0.107 & 0.1007 & 0.2656 & 0.1030 \\
 & SH4 & 0.107 & 0.0955 & 0.2888 & 0.1136 \\
 & MB  & 0.107 & 0.1982 & 0.2814 & 0.1204 \\
\hline
\end{tabular}
\end{table}

Three main conclusions can be drawn from Table~\ref{tab:imse}.

First, the sharpening method reduces the minimum IMSE on all three test functions. This improvement is most obvious for $g_1$. The minimum IMSE of LL is $1.19\times10^{-4}$, while SH3 reduces it to $5.64\times10^{-6}$, a decrease of more than $95\%$. For $g_2$, SH3 reduces IMSE from $0.0248$ to $0.0153$. For $g_3$, the minimum IMSE is obtained from SH1, with a value of $0.0923$, while the result of LL is $0.1753$. These results show that, compared with ordinary local linear estimates, sharpening can bring a clear improvement.

Second, our method is generally comparable to the MB method, and performs better in many cases. For each test function, at least one sharpened estimator has a smaller IMSE than MB. This difference is especially obvious on $g_2$. The minimum IMSE of MB is $0.0300$, which is higher than LL's $0.0248$ and significantly higher than SH3's $0.0153$. In addition, estimates obtained using larger bandwidths may also be more susceptible to finite-sample effects. This may explain why MB is not better than LL in the example of $g_2$.

Third, the optimal sharpening order depends on the specific test function. For $g_1$ and $g_2$, the minimum IMSE is obtained at $l=3$; for $g_3$, it is obtained at $l=1$. When the sharpening order exceeds the optimal value, although increasing the order further can still reduce the bias, the increase in the standard deviation exceeds the benefit from the bias reduction, so the IMSE rebounds slightly. Therefore, a higher sharpening order does not necessarily give a smaller IMSE. As $l$ increases, the selected bandwidth usually increases as well. This shows that once the bias has been reduced, a larger bandwidth can be used to lower the variance.

Figure~\ref{fig:imse} gives the complete IMSE curves. For $g_1$, all methods show an obvious U-shaped pattern. As the sharpening order increases, the lowest point of the curve usually moves to a larger bandwidth. For $g_2$ and $g_3$, when $h$ is larger, the curves of different methods are closer. However, in the smaller and medium bandwidth ranges, the difference between the curves is more obvious, because the influence of bias is greater in this range.

\begin{figure}[ht]
\centering
\includegraphics[width=\linewidth]{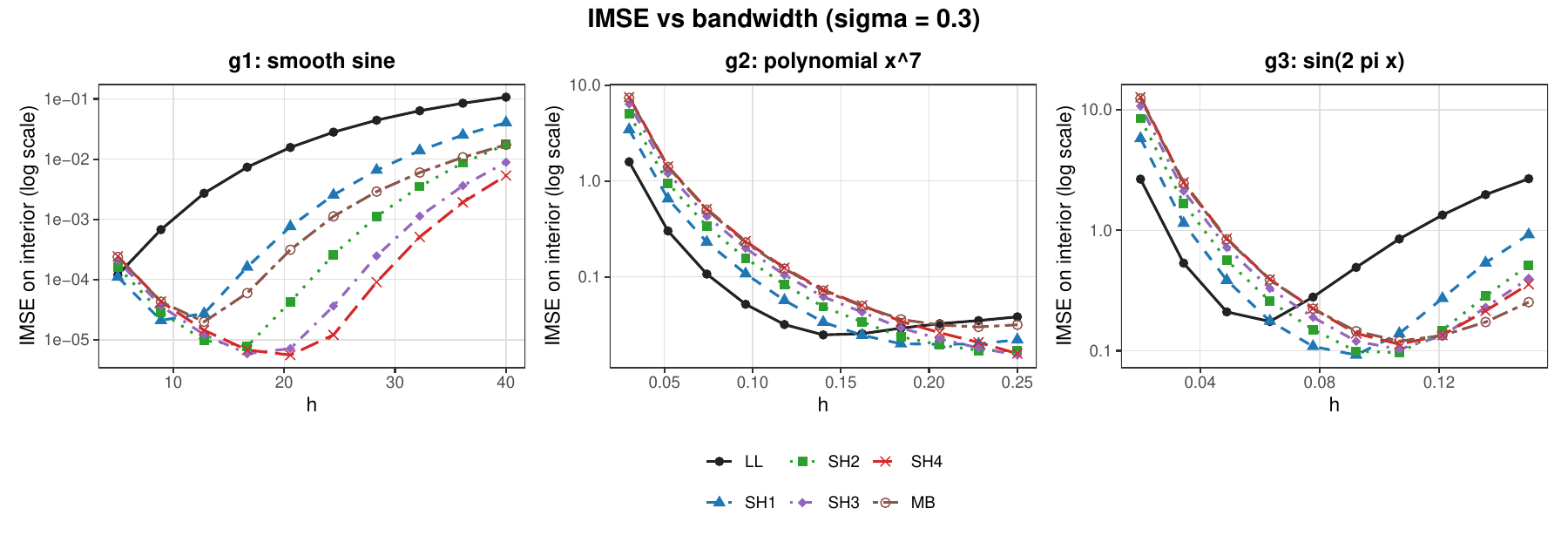}
\caption{IMSE as a function of bandwidth, with $\sigma = 0.3$ and $M = 500$ repetitions.}
\label{fig:imse}
\end{figure}

\subsection{Application to the motorcycle impact data}
\label{sec:mcycle}

Next, we apply the proposed method to a real data set, namely the motorcycle impact data in the \texttt{MASS} package. The data record the head acceleration corresponding to different time points in the simulated motorcycle collision experiment, in units of $g$. After sorting the observations according to the collision time, let $t$ denote the time after the collision and $Y$ the head acceleration. Therefore, $g(t)$ represents the mean acceleration curve, and $g'(t)$ represents the rate of change of the mean acceleration over time.

The simulation results in Section~\ref{sec:finite-sample} show that a higher sharpening order can reduce the bias, but may also increase the variability of the estimate. Therefore, in this real-data application, we mainly consider the lower-order sharpened estimates, namely $l=1$ and $l=2$. Since the true derivative in the real data is unknown, this is also a relatively conservative choice. The motorcycle data contain faster local changes during the main impact stage, so this example can help us examine how low-order sharpening affects these local features. We continue to use the Gaussian kernel. The local linear smoothing matrix is constructed at the observed time points, and the derivative is estimated at 400 equally spaced grid points between the minimum and maximum observation times. We compare the ordinary local linear derivative estimate with the first-order and second-order sharpened estimates and the derivative-adapted multi-bandwidth benchmark described in Section~\ref{theory}. Four bandwidths are considered, namely $h=4,5,6,$ and $7$.

Figure~\ref{fig:mcycle_mean} shows the raw data together with the local linear mean estimate at $h=5$. The fitted curve drops quickly during the early collision stage and then recovers. These changes are clearer on the derivative scale.

\begin{figure}[ht]
\centering
\includegraphics[width=0.80\textwidth]{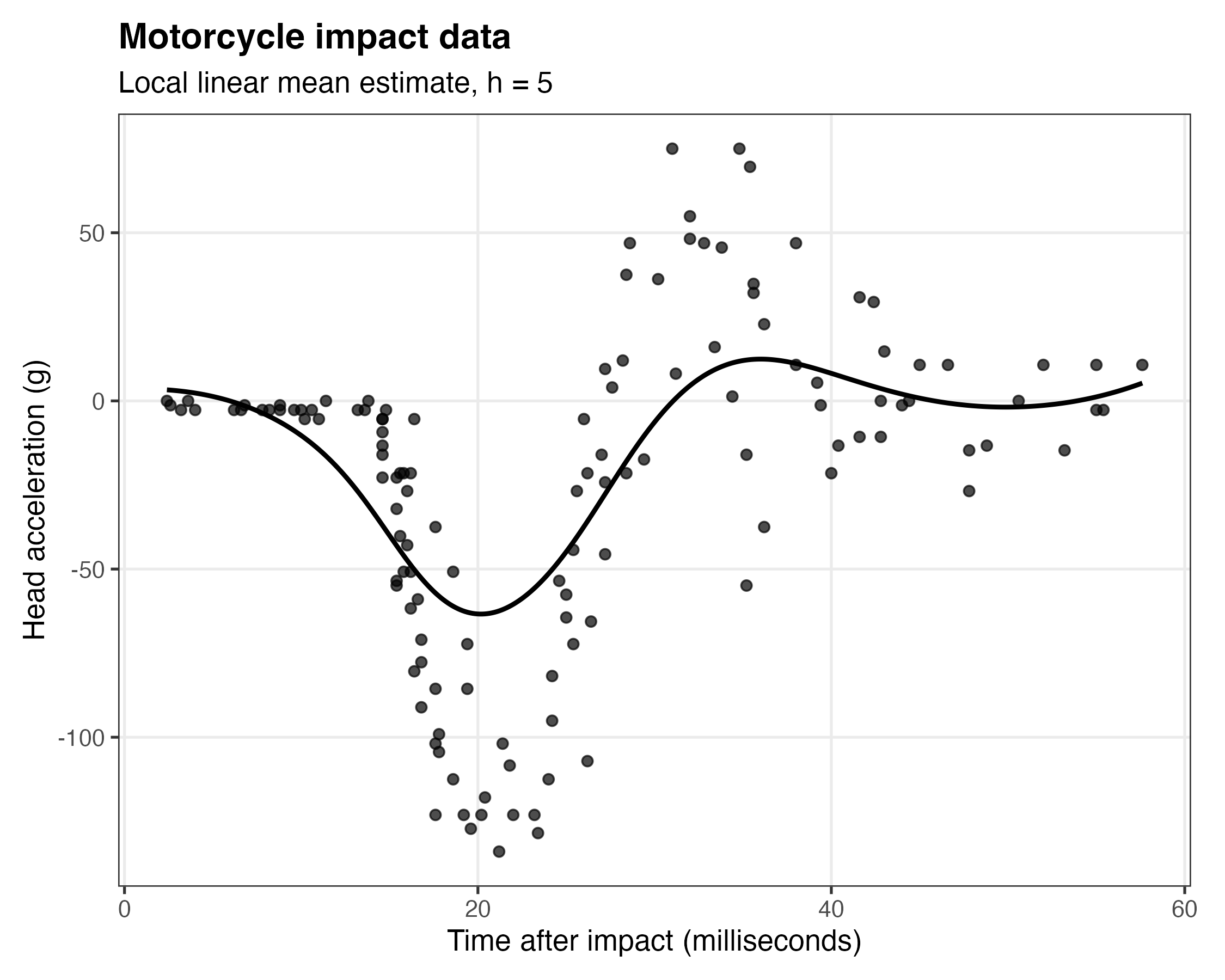}
\caption{Motorcycle impact data and the local linear mean estimate obtained with bandwidth $h=5$.}
\label{fig:mcycle_mean}
\end{figure}

Figure~\ref{fig:mcycle_derivative} compares the four derivative estimators under the selected bandwidths. Compared with the ordinary estimate, the sharpened estimates give a more negative derivative during the initial decline and a larger positive derivative during the recovery, which makes the peaks and valleys of the curve more pronounced.

\begin{figure}[p]
\centering
\includegraphics[width=0.90\textwidth]{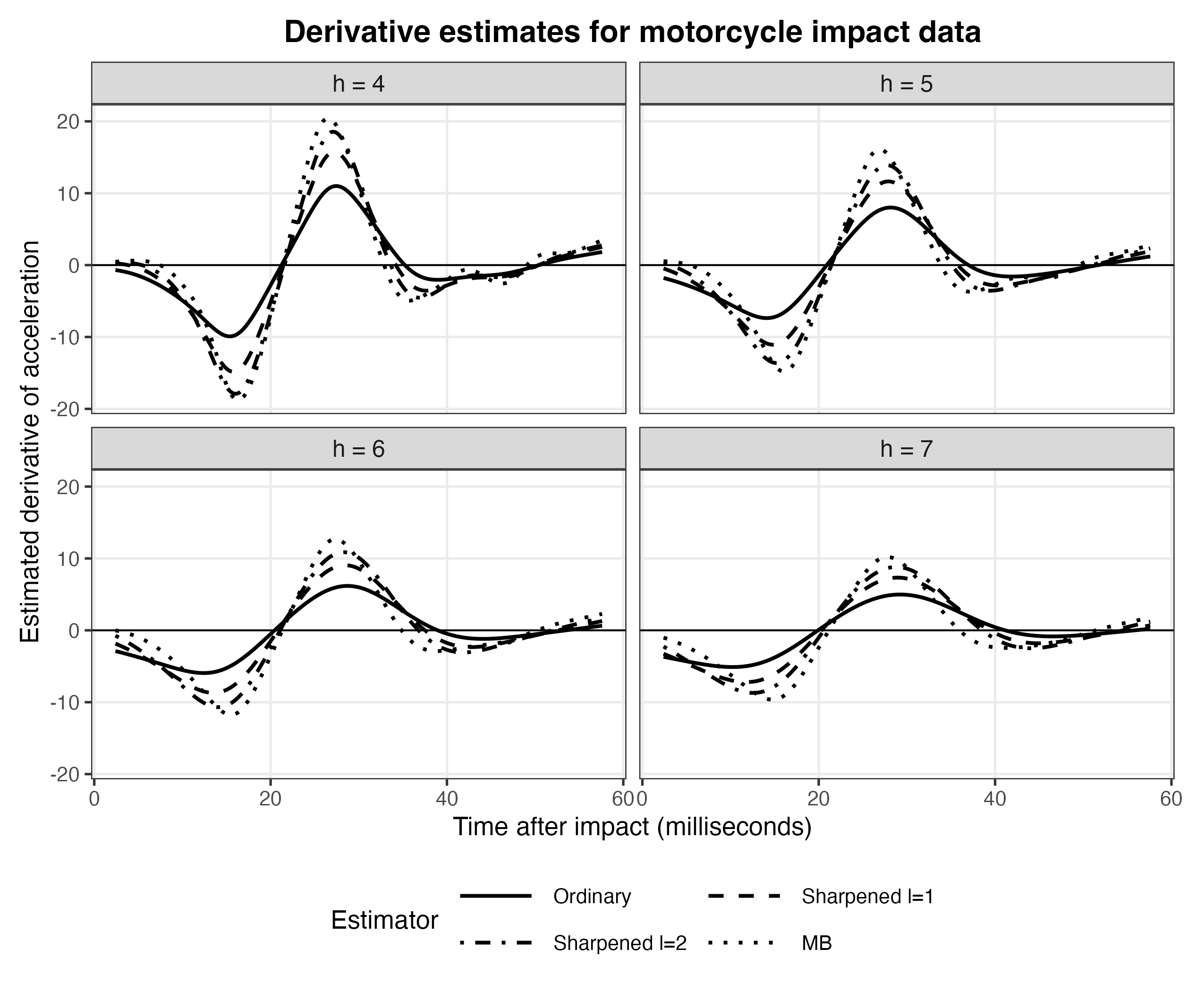}
\caption{Derivative estimates for the motorcycle impact data under bandwidths $h=4,5,6,$ and $7$. The ordinary local linear derivative estimator is compared with the first- and second-order sharpened estimators and the derivative-adapted multi-bandwidth benchmark. The horizontal line represents zero derivative.}
\label{fig:mcycle_derivative}
\end{figure}

Table~\ref{tab:mcycle_summary_h5} reports a few features of the estimated derivative curve at $h=5$. The minimum and maximum derivatives describe the negative and positive phases. The first zero crossing is the estimated time at which the acceleration curve turns from decreasing to increasing. The roughness statistic
\[
\sum_j\left\{\hat g'(t_{j+1})-\hat g'(t_j)\right\}^2
\]
measures the local variation of the curve.

\begin{table}[ht]
\centering
\caption{Summary of the estimated derivative curves for the motorcycle
impact data at bandwidth $h=5$.}
\label{tab:mcycle_summary_h5}
\begin{tabular}{lrrrr}
\hline
Estimator & Min. derivative & Max. derivative &
First zero crossing & Roughness \\
\hline
Ordinary & $-7.372$  & $8.010$  & $20.910$ & $4.461$  \\
Sharpened ($l=1$) & $-11.076$ & $11.643$ & $21.246$ & $11.105$ \\
Sharpened ($l=2$) & $-13.619$ & $13.905$ & $21.416$ & $17.372$ \\
MB & $-14.815$ & $16.116$ & $21.428$ & $24.208$ \\
\hline
\end{tabular}
\end{table}

At $h=5$, the negative and positive phases become more pronounced as the sharpening order increases, but the estimated curves also get rougher. The roughness rises from $4.461$ for the ordinary estimator to $11.105$ and $17.372$ for the first- and second-order sharpened estimators. The MB estimate is the roughest, at $24.208$. This matches what we saw in the simulations, where stronger bias correction came with more variability. 

\section{Conclusion and Future Directions}\label{sec:conc}

This paper proposes an iterative data sharpening method. The purpose is to reduce the bias of the local linear first-order derivative estimate. The method is based on two expectation operators $L_0$ and $L_1$, and the residual operator $R=I-L_0$. By choosing the appropriate sharpening coefficients, we can gradually eliminate the dominant derivative bias. For general symmetric kernels, after $l$ sharpening steps, the order of the bias is $O(h^{2l+2})$. The Gaussian kernel is a special case. At this time, all sharpening coefficients are equal to one, so the estimate becomes simple and only a single bandwidth is used. We have also shown that the multi-bandwidth procedure of \citet{cheng2018bias} can be generalized to derivative estimation.

The simulation results support the claim that sharpening reduces bias. At the same time, the results also show a clear bias--variance trade-off in finite samples. For the smooth sine function, higher-order sharpening can reduce both the average absolute bias and the IMSE. But the results for the higher-order polynomial function are different, because a higher sharpening order can reduce the bias but also increases the standard deviation. Therefore, a moderate sharpening order is usually a sensible choice. The motorcycle impact data show a similar phenomenon. After sharpening, the changes in the estimated derivative become more pronounced, but the curve also becomes rougher. Overall, as a bias reduction method, the practical performance of the method proposed in this paper may depend on the bandwidth, the sample size, the shape of the true function, and the sharpening order.

The theory in this paper mainly focuses on bias reduction at interior points. Boundary effects and the variance properties of the sharpened estimator still need further study. Therefore, future work can continue in several directions. First, one can study sharpening methods with boundary correction, which can be used for derivative estimation near the two ends of the support. Second, the variance properties of the sharpened derivative estimator also need a more detailed discussion. Finally, the method can also be applied to other problems where derivative estimation is directly needed.

\section*{Disclosure statement}

No potential conflict of interest was reported by the authors.

\section*{Funding}

No funding was received.

\bibliography{mybibfile}

\appendix

\section{Additional Proofs}

\subsection{A combinatorial derivation for the Gaussian kernel}\label{app:gaussian_combinatorial}

This appendix gives a direct combinatorial derivation of the explicit multi-sum representation for the iterated sharpening estimator in the Gaussian kernel case.

For the Gaussian kernel,
\[
\mu_{2k}=\frac{(2k)!}{2^k k!},
\]
so the operators \(L_0\) and \(L_1\) admit the expansions
\[
L_0[g](x)=\sum_{k=0}^{\infty}\frac{h^{2k}}{2^k k!}g^{(2k)}(x),
\qquad
L_1[g](x)=\sum_{k=0}^{\infty}\frac{h^{2k}}{2^k k!}g^{(2k+1)}(x).
\]

The iterated sharpened sequence is
\[
g_0:=g,
\qquad
g_l:=g_{l-1}+(I-L_0)^l g,
\qquad l\ge1,
\]
and the binomial expansion
\[
(I-L_0)^l=\sum_{j=0}^l (-1)^j \binom{l}{j}L_0^j
\]
rewrites the update step as
\begin{equation}
g_l
=
g_{l-1}+\sum_{j=0}^l (-1)^j \binom{l}{j}L_0^j g.
\label{eq:appendix_update}
\end{equation}
The corresponding derivative estimator is
\[
\hat g_l'(x):=L_1[g_l](x).
\]

The explicit formula is the following.

\begin{proposition}\label{prop:appendix_gaussian_formula}
For every \(l\ge1\),
\begin{equation}
L_1[g_l](x)
=
g'(x)
+
(-1)^l
\sum_{k_l=1}^{\infty}\cdots\sum_{k_0=1}^{\infty}
\frac{h^{2\sum_{j=0}^l k_j}}
{2^{\sum_{j=0}^l k_j}\,k_0!\cdots k_l!}
\,g^{\left(2\sum_{j=0}^l k_j+1\right)}(x).
\label{eq:appendix_main_formula}
\end{equation}
\end{proposition}

\begin{proof}
Induction on \(l\).

At \(l=1\),
\[
g_1=g+(I-L_0)g,
\]
hence
\[
L_1[g_1]
=
L_1[g]+L_1[(I-L_0)g].
\]
For the Gaussian kernel, \(L_0\) acts formally as \(\exp\{(h^2/2)D_x^2\}\) and \(L_1\) as \(D_xL_0=L_0D_x\), with \(D_x=d/dx\), so the two operators commute. The display above becomes
\[
L_1[g_1]
=
L_1[g]+(I-L_0)L_1[g].
\]
Writing each term out,
\[
L_1[g](x)
=
\sum_{k=0}^{\infty}\frac{h^{2k}}{2^k k!}g^{(2k+1)}(x),
\]
and
\[
L_0L_1[g](x)
=
\sum_{k_0=0}^{\infty}\sum_{k_1=0}^{\infty}
\frac{h^{2(k_0+k_1)}}
{2^{k_0+k_1}k_0!k_1!}
g^{(2(k_0+k_1)+1)}(x),
\]
so that
\[
\begin{aligned}
L_1[g_1](x)
&=
\sum_{k=0}^{\infty}\frac{h^{2k}}{2^k k!}g^{(2k+1)}(x)
+
\sum_{k=0}^{\infty}\frac{h^{2k}}{2^k k!}g^{(2k+1)}(x) \\
&\qquad
-
\sum_{k_0=0}^{\infty}\sum_{k_1=0}^{\infty}
\frac{h^{2(k_0+k_1)}}
{2^{k_0+k_1}k_0!k_1!}
g^{(2(k_0+k_1)+1)}(x).
\end{aligned}
\]
Pulling out the \(k=0\) term gives
\[
L_1[g_1](x)
=
g'(x)
+
\sum_{k=1}^{\infty}\frac{h^{2k}}{2^k k!}g^{(2k+1)}(x)
-
\sum_{k_0=0}^{\infty}\sum_{k_1=1}^{\infty}
\frac{h^{2(k_0+k_1)}}
{2^{k_0+k_1}k_0!k_1!}
g^{(2(k_0+k_1)+1)}(x),
\]
and after the obvious cancellation,
\[
L_1[g_1](x)
=
g'(x)
-
\sum_{k_0=1}^{\infty}\sum_{k_1=1}^{\infty}
\frac{h^{2(k_0+k_1)}}
{2^{k_0+k_1}k_0!k_1!}
g^{(2(k_0+k_1)+1)}(x),
\]
which is \eqref{eq:appendix_main_formula} at \(l=1\).

For the induction step, take \(l\ge1\) and assume
\begin{equation}
L_1[g_l](x)
=
g'(x)
+
(-1)^l
\sum_{k_l=1}^{\infty}\cdots\sum_{k_0=1}^{\infty}
\frac{h^{2\sum_{j=0}^l k_j}}
{2^{\sum_{j=0}^l k_j}\,k_0!\cdots k_l!}
\,g^{\left(2\sum_{j=0}^l k_j+1\right)}(x).
\label{eq:appendix_induction}
\end{equation}
From \eqref{eq:appendix_update},
\[
g_{l+1}=g_l+(I-L_0)^{l+1}g,
\]
so
\[
L_1[g_{l+1}]
=
L_1[g_l]+L_1[(I-L_0)^{l+1}g].
\]
The binomial expansion gives
\[
(I-L_0)^{l+1}
=
\sum_{j=0}^{l+1}(-1)^j\binom{l+1}{j}L_0^j,
\]
and therefore
\[
L_1[(I-L_0)^{l+1}g]
=
\sum_{j=0}^{l+1}(-1)^j\binom{l+1}{j}L_1L_0^j g.
\]
For each fixed \(j\),
\[
L_1L_0^j g
=
\sum_{k_j=0}^{\infty}\cdots\sum_{k_0=0}^{\infty}
\frac{h^{2\sum_{r=0}^j k_r}}
{2^{\sum_{r=0}^j k_r}\,k_0!\cdots k_j!}
\,g^{\left(2\sum_{r=0}^j k_r+1\right)}(x).
\]
Putting the pieces together,
\[
\begin{aligned}
L_1[g_{l+1}](x)
&=
g'(x)
+
(-1)^l
\sum_{k_l=1}^{\infty}\cdots\sum_{k_0=1}^{\infty}
\frac{h^{2\sum_{r=0}^l k_r}}
{2^{\sum_{r=0}^l k_r}\,k_0!\cdots k_l!}
\,g^{\left(2\sum_{r=0}^l k_r+1\right)}(x) \\
&\quad
+
\sum_{j=0}^{l+1}(-1)^j\binom{l+1}{j}
\sum_{k_j=0}^{\infty}\cdots\sum_{k_0=0}^{\infty}
\frac{h^{2\sum_{r=0}^j k_r}}
{2^{\sum_{r=0}^j k_r}\,k_0!\cdots k_j!}
\,g^{\left(2\sum_{r=0}^j k_r+1\right)}(x).
\end{aligned}
\]
The second piece can be rewritten as
\[
\sum_{j=0}^{l+1}(-1)^j\binom{l+1}{j}
\left\{
\sum_{i=-1}^{j}
\binom{j+1}{i+1}X_i
\right\},
\]
where for each \(i\ge0\),
\[
X_i
:=
\sum_{k_i=1}^{\infty}\cdots\sum_{k_0=1}^{\infty}
\frac{h^{2\sum_{r=0}^i k_r}}
{2^{\sum_{r=0}^i k_r}\,k_0!\cdots k_i!}
\,g^{\left(2\sum_{r=0}^i k_r+1\right)}(x).
\]
Using the Pascal identity
\[
\binom{j+1}{i+1}=\binom{j}{i}+\binom{j}{i+1},
\]
we get
\begin{equation}
\sum_{j=0}^{l+1}(-1)^j\binom{l+1}{j}\binom{j+1}{i+1}X_i
=
\sum_{j=0}^{l}(-1)^j\binom{l+1}{j}\binom{j}{i}X_i
+
\sum_{j=0}^{l}(-1)^j\binom{l+1}{j}\binom{j}{i+1}X_i.
\label{eq:appendix_binomial_split}
\end{equation}
The Riordan identity
\[
\sum_{j=0}^{n}(-1)^j\binom{n}{j}\binom{j}{k}=0,
\qquad k<n,
\]
then says that, for each fixed $i$,
\[
\sum_{j=0}^{l}(-1)^j\binom{l+1}{j}\binom{j}{i}X_i
=
\begin{cases}
(-1)^l X_l, & i=l,\\[0.3em]
0, & i<l,
\end{cases}
\]
and
\[
\sum_{j=0}^{l}(-1)^j\binom{l+1}{j}\binom{j}{i+1}X_i
=
\begin{cases}
-(-1)^l X_l+(-1)^{l+1}X_{l+1}, & i=l,\\[0.3em]
0, & i<l.
\end{cases}
\]
The only piece that survives the cancellation is
\[
(-1)^{l+1}X_{l+1}.
\]
Plugging this back in,
\[
L_1[g_{l+1}](x)
=
g'(x)
+
(-1)^{l+1}
\sum_{k_{l+1}=1}^{\infty}\cdots\sum_{k_0=1}^{\infty}
\frac{h^{2\sum_{r=0}^{l+1} k_r}}
{2^{\sum_{r=0}^{l+1} k_r}\,k_0!\cdots k_{l+1}!}
\,g^{\left(2\sum_{r=0}^{l+1} k_r+1\right)}(x),
\]
which is \eqref{eq:appendix_main_formula} with \(l\) replaced by \(l+1\). This finishes the induction proof.
\end{proof}

\end{document}